\documentclass[runningheads,a4paper]{llncs}
\usepackage{amssymb}
\usepackage{graphicx}
\usepackage{amsfonts,amsmath,amssymb}
\usepackage{algorithmicx}
\usepackage{algorithm}
\usepackage{algpseudocode}
\usepackage{enumerate}
\usepackage{color}
\usepackage{caption}

\usepackage{amsthm}
\theoremstyle{definition}
\newtheorem*{runningexample}{Example}

\usepackage{url}
\urldef{\mailsa}\path|{fan, osimpson}@ucsd.edu|    
\newcommand{\keywords}[1]{\par\addvspace\baselineskip
\noindent\keywordname\enspace\ignorespaces#1}

\newcommand{\norm}[1]{||#1||}
\newcommand{\onenorm}[1]{||#1||_1}

\newcommand{\R}{\mathbb{R}}
\newcommand{\E}{\mathbb{E}}
\newcommand{\Var}{\textmd{Var}}
\renewcommand{\P}{\mathrm{Pr}}

\newcommand{\supp}{\textmd{supp}}
\newcommand{\G}{\mathcal{G}}
\renewcommand{\L}{\mathcal{L}}
\renewcommand{\H}{\mathcal{H}}
\newcommand{\dirH}{\H_{S,t}}
\newcommand{\dirh}{H_{S,t}}

\newcommand{\dirhkpr}{\rho_{S,t,f}}
\newcommand{\dirhkapprox}{\hat{\rho}_{S,t,f}}

\newcommand{\localsolver}{\textsc{Local Linear Solver}}
\newcommand{\dirhkpralg}{\texttt{ApproxDirHKPR}}
\newcommand{\dirhkpralgparams}{\texttt{ApproxDirHKPR($G,t,f,S,\epsilon$)}}
\newcommand{\rparamdir}{\frac{16}{\epsilon^3}\log n}
\newcommand{\K}{t/\epsilon}
\newcommand{\dirhkprcomplexity}{O\Big( \epsilon^{-4} t\log n \Big)}
\newcommand{\solverapproxhkpr}{\texttt{SolverApproxDirHKPR}}
\newcommand{\solverapproxhkprparams}{\texttt{SolverApproxDirHKPR($G,t,f,S,\epsilon$)}}
\newcommand{\greensalg}{\texttt{GreensSolver}}
\newcommand{\greensalgparams}{\texttt{GreensSolver($G,b,S,\gamma,\epsilon$)}}
\newcommand{\rparamsolver}{\gamma^{-2} \log(s\gamma^{-1})}
\newcommand{\tparamsolver}{s^3 \log(s^3\gamma^{-1})}
\newcommand{\greenscomplexity}{O\left( \gamma^{-2} \epsilon^{-3} s^3 \log^2(s^3 \gamma^{-1})
\log n \right)}

\begin{document}

%=====================================
%LLNCS formatting
\mainmatter  % start of an individual contribution

% first the title is needed
\title{Solving Local Linear Systems with Boundary Conditions Using
Heat Kernel Pagerank\footnote{An extended abstract appeared in \emph{Proceedings of WAW} (2013)
\cite{cs:hklinear:13}.}}

% a short form should be given in case it is too long for the running head
\titlerunning{Solving Linear Systems with Heat Kernel}

% the name(s) of the author(s) follow(s) next
\author{Fan Chung%
\and Olivia Simpson}

\authorrunning{Chung and Simpson}

% the affiliations are given next; don't give your e-mail address
% unless you accept that it will be published
\institute{Department of Computer Science and Engineering,\\
University of California, San Diego\\
La Jolla, CA 92093\\
\mailsa}

\toctitle{Solving Linear Systems With Heat Kernel}
\tocauthor{Chung and Simpson}
\maketitle

\begin{abstract}
We present an efficient algorithm for solving local linear systems with a
boundary condition using the Green's function of a connected induced subgraph
related to the system.  We introduce the method of using the Dirichlet heat
kernel pagerank vector to approximate local solutions to linear systems in the
graph Laplacian satisfying given boundary conditions over a particular subset of
vertices.  With an efficient algorithm for approximating Dirichlet heat kernel
pagerank, our local linear solver algorithm computes an approximate local
solution with multiplicative and additive error $\epsilon$ by performing
$O(\epsilon^{-5}s^3\log(s^3\epsilon^{-1})\log n)$ random walk steps, where $n$
is the number of vertices in the full graph and $s$ is the size of the local
system on the induced subgraph.

\keywords{local algorithms, graph Laplacian, heat kernel pagerank, symmetric
diagonally dominant linear systems, boundary conditions}
\end{abstract}
%=====================================

\section{Introduction}
There are a number of linear systems which model flow over vertices of
a graph with a given boundary condition.  A classical example is the case of an
electrical network.  Flow can be captured by measuring electric current between
points in the network, and the amount that is injected and removed from the
system.  Here, the points at which voltage potential is measured can be
represented by vertices in a graph, and edges are associated to the ease with
which current passes between two points.  The injection and extraction points
can be viewed as the boundary of the system, and the relationship of the flow
and voltage can be evaluated by solving a system of linear equations over the
measurement points.

Another example is a decision-making process among a network of agents.  Each
agent decides on a value, but may be influenced by the decision of other agents
in the network.  Over time, the goal is to reach consensus among all the agents,
in which each agrees on a common value.  Agents are represented by vertices, and
each vertex has an associated value.  The amount of influence an agent has on a
fellow agent is modeled by a weighted edge between the two representative
vertices, and the communication dynamics can be modeled by a linear system.  In
this case, some special agents which make their own decisions can be viewed as
the boundary.

In both these cases, the linear systems are equations formulated in the graph
Laplacian.  Spectral properties of the Laplacian are closely related to
reachability and the rate of diffusion across vertices in a graph \cite{ch0}.
Laplacian systems  have been used to concisely characterize qualities such as
edge resistance and the influence of communication on edges
\cite{spielman:algorithms:10}.  There is a substantial body of work on efficient
and nearly-linear time solvers for Laplacian linear systems
(\cite{forsythematrixmonte,st:nearlylinear:04,vaidya:preconditioners:91,kmp:optimalsdd:10,kmp:mlognsdd:11,lee2013efficient,kosz:combosdd:13,km:parallel:07,bgkmp:parallel:11,ps:parallelsdd:13,sachdeva2013matrix,ckmpprx:sdd:stoc14}, see also~\cite{vishnoi:lxb}).

The focus of this paper is a localized version of a Laplacian linear solver.  In
a large network, possibly of hundreds of millions of vertices, the algorithms we
are dealing with and the solutions we are seeking are usually of finite support.
Here, by finite we mean the support size depends only on the requested output
and is independent of the full size of the network.  Sometimes we allow sizes up
to a factor of $\log(n)$, where $n$ is the size of the network. 

The setup is a graph and a boundary condition given by a vector with specified
limited support over the vertices.  In the local setting, rather than computing
the full solution we compute the solution over a fraction of the graph and de
facto ignore the vertices with solution values below the multiplicative/additive
error bound.  In essence we avoid computing the entire solution by focusing
computation on the subset itself.  In this way, computation depends on the size
of the subset, rather than the size of the full graph.  We distinguish the two
cases as ``global'' and ``local'' linear solvers, respectively.  We remark that
in the case the solution is not ``local,'' for example, if \emph{all} values are
below the error bound, our alogrithm will return the zero vector -- a valid
approximate solution according to our definition of approximation.  

In this paper, we show how local Laplacian linear systems with a boundary
condition can be solved and efficiently approximated by using Dirichlet heat
kernel pagerank, a diffusion process over an induced subgraph.  We will
illustrate the connection between the Dirichlet heat kernel pagerank vector and
the Green's function, or the inverse of a submatrix of the Laplacian determined
by the subset.  We also demonstrate the method of approximation using random
walks.  Our algorithm approximates the solution to the system restricted to the
subset $S$ by performing $\greenscomplexity$ random walk steps, where $\gamma$
is the error bound for the solver and $\epsilon$ is the error bound for
Dirichlet heat kernel pagerank approximation, and $s$ denotes the size of $S$.
We assume that performing a random walk step and drawing from a distribution
with finite support require constant time.  With this, our algorithm runs in
time $O\left( \gamma^{-2}\epsilon^{-3} \log^4(n) \log^2(\gamma^{-1} \log^3(n))
\right)$ when the support size of the solution is $O(\log n)$.  Note that in our
computation, we do not intend to compute or approximate the matrix form of the
inverse of the Laplacian.  We intend to compute an approximate local solution
which is optimal subject to the (relaxed) definition of approximation.

\subsection{A Summary of the Main Results}
\label{sec:mainresult}
We give an algorithm called \localsolver for approximating a local solution of a
Laplacian linear system with a boundary condition.  The algorithm uses the
connection between the inverse of the restricted Laplacian and the Dirichlet
heat kernel of the graph for approximating the local solution with a sampling of
Dirichlet heat kernel pagerank vectors (heat kernel pagerank restricted to a
subset $S$).  It is shown in Theorem~\ref{thm:localsolver} that the output of
\localsolver~approximates the exact local solution $x_S$ with absolute error
$O(\gamma\norm{b} + \norm{x_S})$ for boundary vector $b$ with probability at
least $1-\gamma$.

We present an efficient algorithm for approximating Dirichlet heat kernel
pagerank vectors, \dirhkpralg.  The algorithm is an extension of the algorithm
in~\cite{cs:iwocahkpr}. The definition of $\epsilon$-approximate vectors is
given in Section~\ref{sec:hkprapprox}.  We note that this notion of
approximation is weaker than the classical notions of total variation distance
among others.  Nevertheless, this ``relaxed'' notion of approximation is used in
analyzing PageRank algorithms (see \cite{bbct:sublinearpr:waw12}, for example)
for massive networks.
 
The full algorithm for approximating a local linear solution, \greensalg, is
presented in Section~\ref{sec:efficientsolver}.  The algorithm is an invocation
of \localsolver~with the \dirhkpralg~called as a subroutine.  The full agorithm
requires $\greenscomplexity$ random walk steps by using the algorithm
\dirhkpralg~with a slight modification.  Our algorithm achieves sublinear time
after preprocessing which depends on the size of the support of the boundary
condition.  The error is similar to the error of \dirhkpralg.

\medskip It is worth pointing out a number of ways our methods can be
generalized.  First, we focus on unweighted graphs, though extending our results
to graphs with edge weights follows easily with a weighted version of the
Laplacian.  Second, we require the induced subgraph on the subset $S$ be
connected.  However, if the induced subgraph is not connected the results can be
applied to components separately, so our requirement on connectivity can be
relaxed.  Finally, we restrict our discussion to linear systems in the graph
Laplacian.  However, by using a linear-time transformation due
to~\cite{gmz:parallel:95} for converting a symmetric, diagonally dominant linear
system to a Laplacian linear system, our results apply to a larger class of
linear systems.

\subsection{Organization}
In Section~\ref{sec:prelims}, we give definitions and basic facts for graph
Laplacian and heat kernel.  In Section~\ref{sec:boundarycondition} the problem
is introduced in detail and provides the setting for the local solver.  The
algorithm, \localsolver, is presented in Section~\ref{sec:localsolver}.  After
this, we extend the solver to the full approximation algorithm using approximate
Dirichlet heat kernel pagerank.  In Section~\ref{sec:hkprapprox}, we give the
definition of local approximation and analyze the Dirichlet heat kernel pagerank
approximation algorithm.  In Section~\ref{sec:efficientsolver}, the full
algorithm for computing an approximate local solution to a Laplacian linear
system with a boundary condition, \greensalg, is given.  Finally in
Section~\ref{sec:example} we illustrate the correctness of the algorithm with an
example network and specified boundary condition.  The example demonstrates
visually what a local solution is and how \greensalg~successfully approximates
the solution within the prescribed error bounds when the solution is
sufficiently local.

\section{Basic Definitions and Facts}
\label{sec:prelims}
Let $G$ be a simple graph given by vertex set $V = V(G)$ and edge set $E =
E(G)$.  Let $u \sim v$ denote $\{u,v\}\in E$.  When considering a real vector
$f$ defined over the vertices of $G$, we say $f\in\R^V$ and the \emph{support}
of $f$ is denoted by $\supp(f) = \{v\in V : f(v) \neq 0\}$.  For a subset of
vertices $S\subseteq V$, we say $s=|S|$ is the size of $S$ and use $f\in\R^S$ to
denote vectors defined over $S$.  When considering a real matrix $M$ defined
over $V$, we say $M \in \R^{V\times V}$, and we use $M_S$ to denote the
submatrix of $M$ with rows and columns indexed by vertices in $S$.  Namely,
$M_S\in \R^{S\times S}$.  Similarly, for a vector $f\in\R^V$, we use $f_S$ to
mean the subvector of $f$ with entries indexed by vertices in $S$.  The
\emph{vertex boundary} of $S$ is $\delta(S) = \{ u \in V \setminus S : \{u,v\}
\in E ~~\text{for some}~ v \in S\}$, and the \emph{edge boundary} is
$\partial(S) = \{ \{u,v\}\in E : u\in S, v\notin S \}$.

\subsection{Graph Laplacians and heat kernel}
\label{sec:graphlaplaciansandheatkernel}
For a graph $G$, let $A$ be the indicator adjacency matrix $A \in
\{0,1\}^{V\times V}$ for which $A_{uv} = 1$ if and only if $\{u,v\} \in E$.  The
\emph{degree} of a vertex $v$ is the number of vertices adjacent to it, $d_v =
|\{u \in V | A_{uv} = 1\}|$.  Let $D$ be the diagonal degree matrix with entries
$D_{vv} = d_v$ on the diagonal and zero entries elsewhere.  The \emph{Laplacian}
of a graph is defined  to be $L=D-A$.  The \emph{normalized Laplacian}, $\L =
D^{-1/2}LD^{-1/2}$, is a degree-nomalized formulation of $L$, given by
\begin{equation*}
\L(u,v) =
\begin{cases}
1 &\mbox{ if } u = v,\\
\frac{-1}{\sqrt{d_ud_v}} &\mbox{ if } u\sim v,\\
0 &\mbox{ otherwise.}
\end{cases}
\end{equation*}

Let $P=D^{-1}A$ be the transition probability matrix for a random walk on the
graph.  Namely, if $v$ is a neighbor of $u$, then $P(u,v) = 1/d_u$ denotes the
probability of moving from vertex $u$ to vertex $v$ in a random walk step.
Another related matrix of significance is the \emph{Laplace operator}, $\Delta =
I-P$.  We note that $\L$ is similar to $\Delta$.

The \emph{heat kernel} of a graph is defined for real $t>0$ by
\begin{equation*}
\H_t = e^{-t\L}.
\end{equation*}
Consider a similar matrix, denoted by $H_t = e^{-t\Delta} =
D^{-1/2}\H_tD^{1/2}$.  For a given $t\in\R^+$ and a preference vector
$f\in\R^V$, the \emph{heat kernel pagerank} is defined by
\begin{equation*}
\rho_{t,f} = f^T H_t,
\end{equation*}
where $f^T$ denotes the transpose of $f$.  When $f$ is a probability
distribution on $V$, we can also express the heat kernel pagerank as an
exponential sum of random walks.  Here we follow the notation for random walks
so that a random walk step is by a right multiplication by $P$:
\begin{equation*}
\rho_{t,f} = f^T e^{-t\Delta} = e^{-t} \sum\limits_{k=0}^{\infty} \frac{t^k}{k!}f^T
P^k.
\end{equation*}

\subsection{Laplacian Linear System}
\label{sec:laplacianlinearsystem}
The examples of computing current flow in an electrical network and consensus in
a network of agents typically require solving linear systems with a boundary
condition formulated in the Laplacian $L = D-A$, where $D$ is the diagonal
matrix of vertex degrees and $A$ is the adjacency matrix of the network.  The
problem in the global setting is the solution to $L \mathtt{x} = \mathtt{b}$,
while the solution $\mathtt{x}$ is required to satisfy the boundary condition
$\mathtt{b}$ in the sense that $\mathtt{x}(v) = \mathtt{b}(v)$ for every vertex
$v$ in the support of $\mathtt{b}$.  Because our analysis uses random walks, we
use the normalized Laplacian $\L = D^{-1/2} L D^{-1/2}$.  We note that the
solution $x$ for Laplacian linear equations of the form $\L x = b$ is equivalent
to solving $L\mathtt{x} = \mathtt{b}$ if we take $\mathtt{x} = D^{-1/2}x$ and
$\mathtt{b} = D^{1/2}b$.  Specifically, our local solver computes the solution
$x$ restricted to $S$, denoted $x_S$, and we do this by way of the discrete
Green's function.

\begin{runningexample}
To illustrate the local setting, we expand upon the problem of a network of
decision-making agents.  Consider a communication network of agents in which a
certain subset of agents $f \subset V$ are \emph{followers} and an adjacent
subset $l \subset V\setminus f$ are \emph{leaders} (see
Figure~\ref{fig:consensus}).  Imagine that the decision values of each agent
depend on neighbors as usual, but also that the values of the leaders are fixed
and will not change.  Specifically, let $d_v$ denote the degree of agent $v$, or
the number of adjacent agents in the communication network, and let $x$ be a
vector of decision values of the agents.  Suppose every follower $v_f$
continuously adjusts their decision according to the protocol:
\begin{equation*}
x(v_f) = x(v_f) - \frac{1}{\sqrt{d_{v_f}}}\sum\limits_{u\sim v_f}\frac{x(u)}{\sqrt{d_u}},
\end{equation*}
while every leader $v_l$ remains fixed at $b(v_l)$.  Then the vector of decision
values $x$ is the solution to the system $\L x = b$, where $x$ is required to
satisfy the boundary condition.

In our example, we are interested in computing the decision values of the
followers of the network where the values of the leaders are a fixed boundary
condition, but continue to influence the decisions of the subnetwork of
followers.

\begin{figure}
\centering
\includegraphics[width=0.7\textwidth]{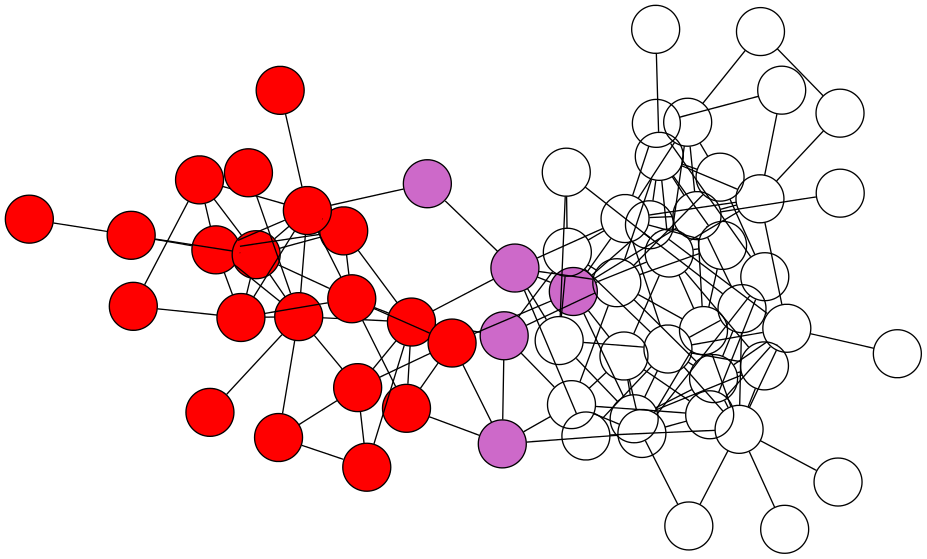}
\caption{A communication network of agents where the leaders (in purple) have
fixed decisions and the followers (in red) compute their decisions based on the
leaders and the subnetwork of followers.  The local solution would be the
decisions of the followers.}
\label{fig:consensus}
\end{figure}
\end{runningexample}

\section{Solving Local Laplacian Linear Systems with a Boundary Condition}
\label{sec:boundarycondition}
For a general connected, simple graph $G$ and a subset of vertices $S$, consider
the linear system $ \L x = b, $ where the vector $b$ has non-empty support on
the vertex boundary of $S$.  The global problem is finding a solution $x$ that
agrees with $b$, in the sense that $x(v) = b(v)$ for every vertex $b$ in the
support of $b$.  In this case we say that $x$ \emph{satisfies the boundary
condition} $b$.

Specifically, for a vector $b\in\R^V$, let $S$ denote a subset of vertices in
the complement of $\supp(b)$.  Then $b$ can be viewed as a function defined on
the vertex boundary $\delta(S)$ of $S$ and we say $b$ is a \emph{boundary
condition} of $S$.  Here we will consider the case that the induced subgraph on
$S$ is connected.

\begin{definition}
\label{def:S}
Let $G$ be a graph and let $b$ be a vector $b\in\R^V$ over the vertices of $G$
with non-empty support.  Then we say a subset of vertices $S\subset
V$ is a \emph{$b$-boundable subset} if
\begin{enumerate}[(i)]
\item $S \subseteq V\setminus \supp(b)$,
\item $\delta(S) \cap \supp(b) \neq \emptyset$,
\item the induced subgraph on $S$ is connected and $\delta(S) \neq
\emptyset$.\label{req:inverse}
\end{enumerate}
\end{definition}
We note that condition (\ref{req:inverse}) is required in our analysis later,
although the general problem of finding a local solution over $S$ can be dealt
with by solving the problem on each connected component of the induced subgraph
on $S$ individually.  We remark that in this setup, we \emph{do not} place any
condition on $b$ beyond having non-empty support.  The entries in $b$ may be
positive or negative.

The global solution to the system $\L x = b$ satisfying the boundary condition
$b$ is a vector $x\in \R^V$ with 
\begin{align}\label{eq:x2}
x(v) =
\begin{cases}
\sum\limits_{u\sim v}\frac{x(u)}{\sqrt{d_vd_u}}  &~\text{if} ~ v \in S\\
b(v) &~\text{if}~ v \not \in S
\end{cases}
\end{align}
for a $b$-boundable subset $S$.  The problem of interest is computing the
\emph{local} solution for the \emph{restriction} of $x$ to the subset $S$,
denoted $x_S$.

The eigenvalues of $\L_S$ are called Dirichlet eigenvalues, denoted $\lambda_1
\leq \lambda_2 \leq \cdots \leq \lambda_s$ where $s = |S|$.  It is easy to check
(see \cite{ch0}) that $0<\lambda_i\leq 2$ since we assume $\delta(S) \neq
\emptyset$.  Thus $\L_S^{-1}$ exists and is well defined.  In fact, $s^{-3} <
\lambda_1 \leq 1$.  

Let $A_{S, \delta S}$ be the  $s \times |\delta(S)|$ matrix by
restricting the columns of $A$ to $\delta(S)$ and rows to $S$.  Requiring $S$ to
be a $b$-boundable subset ensures that the inverse $\L_S^{-1}$
exists~\cite{ch0}.  Then the local solution is described exactly in the
following theorem.

\begin{theorem}
\label{thm:xS}
In a graph $G$, suppose $b$ is a nontrivial vector in $\R^V$ and $S$ is a
$b$-boundable subset.  Then the local solution to the linear system $\L x = b$
satisfying the boundary condition $b$ satisfies
\begin{equation}\label{eq:xS}
x_S = \L_S^{-1}(D_S^{-1/2}A_{S, \delta S}D_{\delta S}^{-1/2}b_{\delta S}).
\end{equation}
\end{theorem}

\begin{proof}
The vector $b_1 := D_S^{-1/2}A_{\delta S}D_{\delta S}^{-1/2}b_{\delta S}$ is
defined over the vertices of $S$, and giveover the vertices of $S$  by
\begin{equation}\label{eq:pfb1}
b_1(v) = \sum\limits_{u\in\delta(S), u\sim v}\frac{b(u)}{\sqrt{d_vd_u}}.
\end{equation}
Also, the vector $\L_Sx_S$ is given by, for $v \in S$,
\begin{equation}\label{eq:pflx}
\L_Sx_S(v) = x(v) - \sum\limits_{u\in S, u\sim v}\frac{x(u)}{\sqrt{d_vd_u}}.
\end{equation}
By (\ref{eq:x2}) and (\ref{eq:xS}), we have
\[
x_S(v) = \sum\limits_{u\in S, u\sim v}\frac{x(u)}{\sqrt{d_vd_u}} + \sum\limits_{u\in\delta(S), u\sim v}\frac{b(u)}{\sqrt{d_vd_u}},
\]
and combining (\ref{eq:pfb1}) and (\ref{eq:pflx}), we have that $x_S = \L_S^{-1}b_1$.
\end{proof}

\subsection{Solving the local system with Green's function}
For the remainder of this paper we are concerned with the local solution $x_S$.
We focus our discussion on the restricted space using the assumptions that the
induced subgraph on $S$ is connected and that $\delta(S)\neq \emptyset$.  In
particular, we consider the \emph{Dirichlet heat kernel}, which is the heat
kernel pagerank restricted to $S$.

The Dirichlet heat kernel is written by $\dirH$ and is defined as $\dirH =
e^{-t\L_S}$.  It is the symmetric version of $\dirh$, where $\dirh =
e^{-t\Delta_S} = D_S^{-1/2}\dirH D_S^{1/2}$.

The spectral decomposition of $\L_S$ is 
\begin{equation*}
\L_S= \sum_{i=1}^s \lambda_i \mathbb{P}_i,
\end{equation*}
where $\mathbb{P}_i$ are the projections to the $i$th orthonormal eigenvectors.
The Dirichlet heat kernel can be expressed as
\begin{equation*}
\dirH = \sum\limits_{i=1}^{s} e^{-t\lambda_i}\mathbb{P}_i.
\end{equation*}

Let $\G$ denote the inverse of $\L_S$.  Namely, $\G \L_S = \L_S \G = I_S$.  Then
\begin{align}
\label{gg}
\G & =  \sum_{i=1}^{s} \frac{1}{\lambda_i} \mathbb{P}_i.
\end{align}
From (\ref{gg}), we see that 
\begin{eqnarray}
\label{ng} \frac 1 2 \leq \norm{\G} \leq \frac 1 {\lambda_1},
\end{eqnarray}
where $\norm{\cdot}$ denotes the spectral norm.  We call $\G$ the \emph{Green's
function}, and $\G$ can be related to $\dirH$ as follows:

\begin{lemma}\label{lem:hkprex}
Let $\G$ be the Green's function of a connected induced subgraph on $S \subset
V$ with $s=|S|$.  Let $\dirH$ be the Dirichlet heat kernel with respect to $S$.
Then
\begin{eqnarray*}
\G  = \int_0^{\infty} \dirH~ \mathrm{d}t.
\end{eqnarray*}
\end{lemma}

\begin{proof}
By our definition of the heat kernel,
\begin{align*}
\int_0^{\infty}\dirH~ \mathrm{d}t 
&= \int_0^{\infty} \Big(\sum\limits_{i=1}^{s} e^{-t\lambda_i}\mathbb{P}_i \Big) \mathrm{d}t\\
&= \sum\limits_{i=1}^{s}
\Big(\int_0^{\infty}e^{-t\lambda_i}\,\mathrm{d}t\Big)\mathbb{P}_i\\
&= \sum\limits_{i=1}^{s} \frac{1}{\lambda_i}\mathbb{P}_i\\
&= \G.
\end{align*}
\end{proof}

Equipped with the Green's function, the solution (\ref{eq:xS}) can be expressed
in terms of the Dirichlet heat kernel.  As a corollary to Theorem~\ref{thm:xS}
we have the following.

\begin{corollary}\label{cor:xSintegral}
In a graph $G$, suppose $b$ is a nontrivial vector in $\R^V$ and $S$ is a
$b$-boundable subset.  Then the local solution to the linear system $\L x = b$
satisfying the boundary condition $b$ can be written as 
\begin{equation}\label{eq:xSintegral}
x_S = \int_0^{\infty}\dirH b_1 ~\mathrm{d}t,
\end{equation}
where $b_1 = D_S^{-1/2}A_{S, \delta S}D_{\delta S}^{-1/2}b_{\delta S}$.
\end{corollary}
The computation of $b_1$ takes time proportional to the size of the edge
boundary.

\section{A Local Linear Solver Algorithm with Heat Kernel Pagerank}
\label{sec:localsolver}
In the previous section, we saw how the local solution $x_S$ to the system
satisfying the boundary condition $b$ can be expressed in terms of integrals of
Dirichlet heat kernel in (\ref{eq:xSintegral}).  In this section, we will show
how these integrals can be well-approximated by sampling a finite number of
values of Dirichlet heat kernel (Theorem~\ref{thm:sampler}) and Dirichlet heat
kernel pagerank (Corollary~\ref{cor:sampler}).  All norms $\norm{\cdot}$ in this
section are the $L_2$ norm.

\begin{theorem}
\label{thm:sampler}
Let $G$ be a graph and $\L$ denote the normalized Laplacian of $G$.  Let $b$ be
a nontrivial vector $b\in \R^V$ and $S$ a $b$-boundable subset, and let $b_1 =
D_S^{-1/2}A_{S, \delta S}D_{\delta S}^{-1/2}b_{\delta S}$.  Then the local
solution $x_S$ to the linear system $\L x = b$ satisfying the boundary condition
$b$ can be computed by sampling $\dirH b_1$ for $r = \rparamsolver$ values.  If
$\hat{x}_S$ is the output of this process, the result has error bounded by
\begin{equation*}
\norm{x_S - \hat{x}_S} = O\big(\gamma(\norm{b_1} + \norm{x_S})\big)
\end{equation*}
with probability at least $1 - \gamma$.
\end{theorem}

We prove Theorem~\ref{thm:sampler} in two steps.  First, we show how the
integral (\ref{eq:xSintegral}) can be expressed as a finite Riemann sum without
incurring much loss of accuracy in Lemma~\ref{lem:sum}.  Second, we show in
Lemma~\ref{lem:sample} how this finite sum can be well-approximated by its
expected value using a concentration inequality.

\begin{lemma}
\label{lem:sum}
Let $x_S$ be the local solution to the linear system $\L x = b$ satisfying
the boundary condition $b$ given in (\ref{eq:xSintegral}).  Then, for $T =
\tparamsolver$ and $N = T/\gamma$, the error incurred by taking a right
Riemann sum is
\begin{equation*}
\norm{x_S - \sum\limits_{j=1}^N \H_{S,jT/N} \frac{T}{N} b_1} \leq 
\gamma(\norm{b_1} + \norm{x_S}),
\end{equation*}
where $b_1 = D_S^{-1/2}A_{S, \delta S}D_{\delta S}^{-1/2}b_{\delta S}$.
\end{lemma}

\begin{proof}
First, we see that:
\begin{align}
\norm{\dirH} &= \norm{\sum\limits_i e^{-t\lambda_i}\mathbb{P}_i}\nonumber\\
&\leq e^{-t\lambda_1}\norm{\sum\limits_i \mathbb{P}_i}\nonumber\\
&= e^{-t\lambda_1}\label{eq:error}
\end{align}
where $\lambda_i$ are Dirichlet eigenvalues for the induced subgraph $S$.  So
the error incurred by taking a definite integral up to $t=T$ to approximate the
inverse is the difference
\begin{align*}
\norm{x_S - \int_0^{T}\dirH b_1~\mathrm{d}t}
&= \norm{\int_T^{\infty}\dirH b_1~\mathrm{d}t}\\
&\leq \int_T^{\infty} e^{-t\lambda_1}\norm{b_1}~\mathrm{d}t\\
&\leq \frac{1}{\lambda_1} e^{-T\lambda_1} \norm{b_1}.
\end{align*}
Then by the assumption on $T$ the error is bounded by $\norm{x_S -
\int_0^{T}\dirH b_1~\mathrm{d}t} \leq \gamma\norm{b_1}$.

Next, we approximate the definite integral in $[0,T]$ by discretizing it.  That
is, for a given $\gamma$, we choose $N= T/\gamma$ and divide the interval
$[0,T]$ into $N$ intervals of size $T/N$.  Then a finite Riemann sum is close to
the definite integral:
\begin{align*}
\norm{\int_0^T \dirH b_1~\mathrm{d}t - \sum\limits_{j=1}^N
\H_{S,jT/N} b_1\frac{T}{N}}
&\leq \gamma \norm{\int_0^T \dirH b_1~\mathrm{d}t}\\
&\leq \gamma\norm{x_S}.
\end{align*}
This gives a total error bounded by $\gamma(\norm{b_1} + \norm{x_S})$.
\end{proof}

\begin{lemma}
\label{lem:sample}
The sum $\sum\limits_{j=1}^N \H_{S,jT/N}b_1 \frac{T}{N}$ can be approximated by
sampling $r = \rparamsolver$ values of $\H_{S,jT/N}b_1$ where $j$ is drawn from
$[1,N]$.  With probability at least $1 - \gamma$, the result has
multiplicative error at most $\gamma$.
\end{lemma}

A main tool in our proof of Lemma~\ref{lem:sample} is the following matrix
concentration inequality (see \cite{cr:spectrarandom:11}, also variations in
\cite{aw}, \cite{cm}, \cite{oli}, \cite{gross}, \cite{tropp}).

\begin{theorem}
\label{thm:concentration}
Let $X_1, X_2, \dots, X_m$ be independent random $n\times n$ Hermitian matrices.
Moreover, assume that $\| X_i-\E(X_i)\|\leq M$  for all $i$, and put
$v^2=\|\sum_i\text{var}(X_i)\|$.  Let $X=\sum_i X_i$.  Then for any $a>0$, 
\begin{equation*}
\P(\|X-\E(X)\|>a)\leq 2n\exp\left(-\frac{a^2}{2v^2+2Ma/3}\right),
\end{equation*}
where $\norm{\cdot}$ denotes the spectral norm.
\end{theorem}

\begin{proof}[Proof of Lemma~\ref{lem:sample}]
Suppose without loss of generality that $\norm{b_1} = 1$.  Let $Y$ be a random
variable that takes on the vector $\H_{S,jT/N} b_1$ for every $j\in[1,N]$ with
probability $1/N$.  Then $\E(Y) = \frac{1}{N} \sum\limits_{j=1}^N \H_{S,jT/N}
b_1$.  Let $X = \sum\limits_{i=1}^r X_j$ where each $X_j$ is a copy of $Y$, so
that $\E(X) = r\E(Y)$.

Now consider $\mathbb{Y}$ to be the random variable that takes on the projection
matrix $\H_{S,jT/N} b_1 (\H_{S, jT/N} b_1)^T$ for every $j\in[1,N]$ with
probability $1/N$, and $\mathbb{X}$ is the sum of $r$ copies of $\mathbb{Y}$.
Then we evaluate the expected value and variance of $\mathbb{X}$ as follows:
\begin{eqnarray*}
\norm{\E(\mathbb{X})}   &=& r\norm{\E(\mathbb{Y})}\\
\norm{\Var(\mathbb{X})} &=& r\norm{\Var(\mathbb{Y})}
\leq \norm{\frac{r}{N}\sum_{j=1}^N \H_{S,jT/N} b_1
(\H_{S,jT/N} b_1)^T\norm{\H_{S,jT/N} b_1}^2}\\
&\leq& r\norm{\E(\mathbb{Y})}.
\end{eqnarray*}
We now apply Theorem~\ref{thm:concentration} to $\mathbb{X}$.
We have
\begin{align*}
\P\big( \norm{\mathbb{X} - \E(\mathbb{X})} \geq \gamma \norm{\E(\mathbb{X})}\big)
&\leq 2s \exp\left( -\frac{\gamma^2 \norm{\E(\mathbb{X})}^2}{2\Var(\mathbb{X}) +
\frac{2\gamma \norm{\E(\mathbb{X})} M} 3} \right)\\
&\leq 2s \exp\left( -\frac{\gamma^2 r^2\norm{\E(\mathbb{Y})} }{r+ 2\gamma r
M/3} \right)\\
%&\leq 2s \exp\left( -\frac{\gamma^2 r^2\norm{\E(\mathbb{Y})} }{r + 2\gamma
%\norm{\E(\mathbb{Y})}/3} \right)\\
&\leq 2s \exp\left(-\frac{\gamma^2 r}{2} \right).
\end{align*}
Therefore we have $\P\big( \norm{\mathbb{X} - \E(\mathbb{X})} \geq \gamma
\norm{\E(\mathbb{X})}\big) \leq \gamma$ if we choose $r \geq \rparamsolver$.
Further, this implies the looser bound:
\begin{equation*}
\P\big( \norm{X - \E(X)} \geq \gamma\norm{\E(X)} \big) \leq \gamma.
\end{equation*}
Then $\E(Y) = \frac{1}{r}\E(X)$ is close to $\frac{1}{r}X$ and
\begin{align*}
\norm{\sum\limits_{j=1}^N \H_{S,jT/N}b_1 \frac{1}{N} - \frac{1}{r} X}
& \leq \gamma \norm{\sum\limits_{j=1}^N \H_{S,jT/N}b_1 \frac{1}{N}}\\
\norm{\sum\limits_{j=1}^N \H_{S,jT/N}b_1\frac{T}{N} - \frac{T}{r} X}
& \leq \gamma \norm{\sum\limits_{j=1}^N \H_{S,jT/N}b_1 \frac{T}{N}}
\end{align*}
with probability at least $1-\gamma$, as claimed.
\end{proof}

\begin{proof}[Proof of Theorem~\ref{thm:sampler}]
Let $X$ be the sum of $r$ samples of $\H_{S,jT/N}b_1$ with $j$ drawn from
$[0,N]$, and let $\hat{x}_S = \frac{T}{r}X$.  Then combining
Lemmas~\ref{lem:sum} and~\ref{lem:sample}, we have
\begin{align*}
\norm{x_S - \hat{x}_S} 
&\leq \gamma\big(\norm{b_1} + \norm{x_S} + \norm{\sum\limits_{j=1}^N
\H_{S,jT/N}b_1 \frac{T}{N}}\big)\\
&\leq O\big(\gamma(\norm{b_1} + \norm{x_S})\big).
\end{align*}
By Lemma~\ref{lem:sample}, this bound holds with probability at least
$1-\gamma$ .
\end{proof}

The above analysis allows us to approximate the solution $x_S$ by sampling
$\H_{S,t}b_1$ for various $t$.  The following corollary is similar to
Theorem~\ref{thm:sampler} except we use the asymmetric version of the Dirichlet
heat kernel which we will need later for using random walks.  In particular, we
use Dirichlet heat kernel pagerank vectors.  Dirichlet heat kernel pagerank is
also defined in terms of a subset $S$ whose induced subgraph is connected, and a
vector $f\in\R^S$ by the following:
\begin{equation}
\label{eq:dirhkpr}
\dirhkpr = f^T\dirh.
\end{equation}

\begin{corollary}
\label{cor:sampler}
Let $G$ be a graph and $\L$ denote the normalized Laplacian of $G$.  Let $b$ be
a nontrivial vector $b\in \R^V$ and $S$ be a $b$-boundable subset.  Let $b_2 =
(D_S^{-1/2}A_{S, \delta S}D_{\delta S}^{-1/2}b_{\delta S})^TD_S^{1/2}$.  Then
the local solution $x_S$ to the linear system $\L x = b$ satisfying the boundary
condition $b$ can be computed by sampling $\rho_{S,t,b_2}$ for $r =
\rparamsolver$ values.  If $\hat{x}_S$ is the output of this process, the result
has error bounded by
\begin{equation*}
\norm{x_S - \hat{x}_S} = O\big(\gamma(\norm{b_1} + \norm{x_S})\big),
\end{equation*}
where $b_1 = D_S^{-1/2}A_{\delta S}D_{\delta S}^{-1/2}b_{\delta S}$, with
probability at least $1 - \gamma$.
\end{corollary}

\begin{proof}
First, we show how $x_S$ can be given in terms of Dirichlet heat kernel
pagerank.
\begin{align*}
x_S^T &= \int_0^{\infty}b_1^T \dirH ~\mathrm{d}t\\
&= \int_0^{\infty}b_1^T (D_S^{1/2}\dirh D_S^{-1/2}) ~\mathrm{d}t\\
&= \int_0^{\infty}b_2 \dirh D_S^{-1/2} ~\mathrm{d}t,
~~\mbox{ where } b_2 = b_1^T D_S^{1/2}\\
&= \int_0^{\infty}\rho_{S,t,b_2} ~\mathrm{d}t ~D_S^{-1/2},
\end{align*}
and we have an expression similar to (\ref{eq:xSintegral}).  Then by
Lemma~\ref{lem:sum}, $x_S^T$ is close to $\sum\limits_{j=1}^N \rho_{S,jT/N,b_2}
\frac{T}{N} D_S^{-1/2}$ with error bounded by $O\big(\gamma(\norm{b_1} +
\norm{x_S})\big)$.  From Lemma~\ref{lem:sample}, this can be approximated to
within $O(\gamma\norm{x_S})$ multiplicative error using $r=\rparamsolver$
samples with probability at least $1-\gamma$.  This gives total additive and
multiplicative error within $O(\gamma)$.
\end{proof}

\subsection{The \localsolver~Algorithm}
We present an algorithm for computing a local solution to a Laplacian linear
system with a boundary condition.

\begin{algorithm}
\floatname{algorithm}{Algorithm}
\caption{\localsolver}
\label{alg:localsolver}
\textbf{input:} graph $G$, boundary vector $b\in\R^V$, subset $S\subset V$, solver error
parameter $0<\gamma<1$.\\
\textbf{output:} an approximate local solution $\mathtt{x}$ with additive and
multiplicative error $\gamma$ to the local system $x_S=\G b_1$ satisyfing the
boundary condition $b$.\\
\begin{algorithmic}[1]
  \State $s\gets |S|$
  \State initialize a $0$-vector $\mathtt{x}$ of dimension $s$
  \State $b_1 \gets D_S^{-1/2}A_{S, \delta S}D_{\delta S}^{-1/2}b_{\delta S}$
  \State $b_2 \gets b_1^TD_S^{1/2}$
  \State $T \gets \tparamsolver$
  \State $N \gets T/\gamma$
  \State $r \gets \rparamsolver$
  \For{$i=1$ to $r$}
    \State draw $j$ from $[1,N]$ uniformly at random
    \State $x_i \gets \rho_{S,jT/N, b_2}$ 
    \State $\mathtt{x} \gets \mathtt{x} + x_i$
  \EndFor\\
  \Return $T/r \cdot \mathtt{x}D_S^{-1/2}$
\end{algorithmic}   
\end{algorithm}

\begin{theorem}
\label{thm:localsolver}
Let $G$ be a graph and $\L$ denote the normalized Laplacian of $G$.  Let $b$ be a
nontrivial vector $b\in \R^V$, $S$ a $b$-boundable subset, and let $b_1 =
D_S^{-1/2}A_{S, \delta S}D_{\delta S}^{-1/2}b_{\delta S}$.  For the linear
system $\L x = b$, the solution $x$ is required to satisfy the boundary
condition $b$, and let $x_S$ be the local solution.  Then the approximate
solution $\mathtt{x}$ output by the \localsolver~algorithm has an error bounded
by
\begin{equation*}
\norm{x_S-\mathtt{x}} = O\big( \gamma(\norm{b_1} + \norm{x_S})
\big)
\end{equation*}
with probability at least $1-\gamma$.
\end{theorem}

\begin{proof}
The correctness of the algorithm follows from Corollary~\ref{cor:sampler}.
\end{proof}

The algorithm involves $r = \rparamsolver$ Dirichlet heat kernel pagerank
computations, so the running time is proportional to the time for computing $b_2
e^{-T\Delta_S}$ for $T=\tparamsolver$.

In the next sections, we discuss an efficient way to approximate a Dirichlet heat
kernel pagerank vector and the resulting algorithm \greensalg~that returns
approximate local solutions in sublinear time.

\section{Dirichlet Heat Kernel Pagerank Approximation Algorithm}
\label{sec:hkprapprox}
The definition of Dirichlet heat kernel pagerank in (\ref{eq:dirhkpr}) is given
in terms of a subset $S$ and a vector $f\in\R^S$.  Our goal is to express this
vector as the stationary distribution of random walks on the graph in order to
design an efficient approximation algorithm.

Dirichlet heat kernel pagerank is defined over the vertices of a subset $S$ as
follows:
\begin{align*}
\dirhkpr &= f^T H_{S,t} = f^T e^{-t\Delta_S} = f^T e^{-t(I_S - P_S)}\\
&= \sum\limits_{k=0}^{\infty} e^{-t} \frac{t^k}{k!}f^T P_S^k.
\end{align*}
That is, it is defined in terms of the transition probability matrix $P_S$ --
the restriction of $P$ where $P$ describes a random walk on the graph.  We can
interpret the matrix $P_S$ as the transition probability matrix of the following
so-called \emph{Dirichlet random walk}: Move from a vertex $u$ in $S$ to a
neighbor $v$ with probability $1/d_u$.  If $v$ is not in $S$, abort the walk and
ignore any probability movement.  Since we only consider the diffusion of
probability within the subset, any random walks which leave $S$ cannot be
allowed to return any probability to $S$.  To prevent this, random walks that do
not remain in $S$ are ignored.

We recall some facts about random walks.  First, if $g$ is a probabilistic
function over the vertices of $G$, then $g^T P^k$ is the probability
distribution over the vertices after performing $k$ random walk steps according
to $P$ starting from vertices drawn from $g$.  Similarly, when $f$ is a
probabilistic fuction over $S$, $f^T P_S^k$ is the distribution after $k$
Dirichlet random walk steps.  Consider a Dirichlet random walk process in which
the number of steps taken, $k$ (where steps are taken according to a Dirichlet
random walk as described above), is a Poisson random variable with mean $t$.
That is, $k$ steps are taken with probability $p_t(k) = e^{-t}\frac{t^k}{k!}$.
Then, the Dirichlet heat kernel pagerank is the expected distribution of this
process.

In order to use random walks for approximating Dirichlet heat kernel pagerank,
we perform some preprocessing for general vectors $f\in\R^S$.  Namely, we do
separate computations for the positive and negative parts of the vector, and
normalize each part to be a probability distribution.

Given a graph and a vector $f\in\R^S$, the algorithm \dirhkpralg~computes
vectors that $\epsilon$-approximate the Dirichlet heat kernel pagerank
$\dirhkpr$ satisfying the following criteria:

\begin{definition}
\label{def:approxhkpr}
Let $G$ be a graph and let $S \subset V$ be a subset of vertices.  Let
$f\in\R^S$ be a probability distribution vector over the vertices of $S$ and let
$\dirhkpr$ be the Dirichlet heat kernel pagerank vector according to $S$, $t$
and $f$.  Then we say that $\nu \in \R^S$ is an \emph{$\epsilon$-approximate
Dirichlet heat kernel pagerank vector} if
\begin{enumerate}
\item for every vertex $v$ in the support of $\nu$, 
% $(1-\epsilon)\dirhkpr(v) -\epsilon \leq \nu(v) \leq (1+\epsilon)\dirhkpr(v),$
$|\dirhkpr(v) - \nu(v)| \leq \epsilon \cdot \dirhkpr(v),$
and
\item for every vertex with $\nu(v) = 0$, it must be that $\dirhkpr(v) \leq \epsilon$.
\end{enumerate}
When $f$ is a general vector, an $\epsilon$-approximate Dirichlet heat kernel
pagerank vector has an additional additive error of $\epsilon\onenorm{f}$ by
scaling, where $\onenorm{\cdot}$ denotes the $L_1$ norm.
\end{definition}
For example, the zero-vector is an $\epsilon$-approximate of any vector with all
entries of value $< \epsilon$.  We remark that for a vector $f$ with $L_1$ norm
$1$, the Dirichlet heat kernel pagerank vector $\dirhkpr$ has at most
$1/\epsilon$ entries with values at least $\epsilon$.  Thus a vector that
$\epsilon$-approximates $\dirhkpr$ has support of size at most $1/\epsilon$.

\begin{algorithm}
\floatname{algorithm}{Algorithm}
\caption{\dirhkpralgparams}
\label{}
\algblock[Name]{Start}{End}
input: a graph $G$, $t\in\R^+$, vector $f\in\R^S$, subset $S\subset V$, error parameter $0 < \epsilon < 1$.\\
output: $\rho$, an $\epsilon$-approximation of $\dirhkpr$.\\
\begin{algorithmic}[1]
  \State $s \gets |S|$
  \State initialize $0$-vector $\rho$ of dimension $s$
  \State $f_+ \gets$ the positive portion of $f$
  \State $f_- \gets$ the negative portion of $f$ so that $f = f_+ - f_-$
  \State $f_+' \gets f_+/\onenorm{f_+}$
  \Comment{\textit{normalize $f_+$ to be a probability distribution vector}}
  \State $f_-' \gets f_-/\onenorm{f_-}$
  \Comment{\textit{normalize $f_-$ to be a probability distribution vector}}
  \State $r \gets \rparamdir$
  \For {$r$ iterations}
    \State choose a starting vertex $u_1$ according to the distribution vector
$f_+'$
    \State $k \sim Poiss(t)$
    \Comment{choose $k$ with probability $e^{-t}\frac{t^k}{k!}$}
    \State $k \gets min\{k, t/\epsilon\}$\label{line:boundk}
    \State simulate $k$ steps of a $P=D^{-1}A$ random walk
    \If {the random walk leaves $S$}:
      \State do nothing for the rest of this iteration
    \Else
      \State let $v_1$ be the last vertex visited in the walk
      \State $\rho[v_1] \gets \rho[v_1] + \onenorm{f_+}$
    \EndIf
    \State choose a starting vertex $u_2$ according to the distribution vector
$f_-'$
    \State $k \sim Poiss(t)$
    \Comment{choose $k$ with probability $e^{-t}\frac{t^k}{k!}$}
    \State $k \gets min\{k, t/\epsilon\}$\label{line:boundk2}
    \State simulate $k$ steps of a $P=D^{-1}A$ random walk
    \If {the random walk leaves $S$}:
      \State do nothing for the rest of this iteration
    \Else
      \State let $v_2$ be the last vertex visited in the walk
      \State $\rho[v_2] \gets \rho[v_2] + \onenorm{g_-}$
    \EndIf
  \EndFor\\
  \State $\rho \gets \rho/r$\\
  \Return $\rho$
\end{algorithmic} 
\end{algorithm}

The time complexity of \dirhkpralg~is given in terms of random walk steps.  As
such, the analysis assumes access to constant-time queries returning (i) the
destination of a random walk step, and (ii) a sample from a distribution.

\begin{theorem}
\label{thm:dirhkpralg}
Let $G$ be a graph and $S$ a proper vertex subset such that the induced subgraph
on $S$ is connected.  Let $f$ be a vector $f\in\R^S$, $t\in\R^+$, and
$0<\epsilon<1$.  Then the algorithm \dirhkpralgparams~ outputs an
$\epsilon$-approximate Dirichlet heat kernel pagerank vector $\dirhkapprox$ with
probability at least $1-\epsilon$.  The running time of \dirhkpralg~ is
$\dirhkprcomplexity$, where the constant hidden in the big-O notation reflects
the time to perform a random walk step.
\end{theorem}

Our analysis relies on the usual Chernoff bounds restated below.  They will be
applied in a similar fashion as in~\cite{bbct:sublinearpr:waw12}.

\begin{lemma}[\cite{bbct:sublinearpr:waw12}]
\label{lem:chernoff}
Let $X_i$ be independent Bernoulli random variables with $X = \sum\limits_{i =
1}^r X_i$.  Then, 
\begin{enumerate}
\item for $0 < \epsilon < 1$, $\P(X < (1-\epsilon)r\E(X)) < \exp(-\frac{\epsilon^2}{2}r\E(X))$
\item for $0 < \epsilon < 1$, $\P(X > (1+\epsilon)r\E(X)) < \exp(-\frac{\epsilon^2}{4}r\E(X))$
\item for $c \geq 1$, $\P(X > (1+c)r\E(X)) < \exp(-\frac{c}{2}r\E(X))$.
\end{enumerate}
\end{lemma}

\begin{proof}[Proof of Theorem~\ref{thm:dirhkpralg}]
For the sake of simplicity, we provide analysis for the positive part of the
vector, $f := f_+$, noting that it is easily applied similarly to the negative
part as well.

The vector $f' = f/\onenorm{f}$ is a probability distribution and the heat
kernel pagerank $\rho'_{S,t,f} = \rho_{S,t,f}/\onenorm{f}$ can be interpreted as
a series of Dirichlet random walks in which, with probability
$e^{-t}\frac{t^k}{k!}$, $f'^TP_S^k$ is contributed to $\rho'_{S,t,f}$.  This is
demonstrated by examining the coefficients of the terms, since
\begin{equation*}
e^{-t} \sum\limits_{k=0}^{\infty} \frac{t^k}{k!} = 1.
\end{equation*}
The probability of taking $k\sim Pois(t)$ steps such
that $k \geq t/\epsilon$ is less than $\epsilon$ by Markov's inequality.
Therefore, enforcing an upper bound of $K = t/\epsilon$ for the number of random
walk steps taken is enough mixing time with probability at least $1-\epsilon$.

For $k\leq \K$, our algorithm approximates $f'^T P_S^k$ by simulating $k$ random
walk steps according to $P$ as long as the random walk remains in $S$.  If the
random walk ever leaves $S$, it is ignored.  To be specific, let $X^v_k$ be the
indicator random variable defined by $X^v_k = 1$ if a random walk beginning from
a vertex $u$ drawn from $f' = f/\onenorm{f}$ ends at vertex $v$ in $k$ steps
without leaving $S$.  Let $X^v$ be the random variable that considers the random
walk process ending at vertex $v$ in \emph{at most} $k$ steps without leaving
$S$.  That is, $X^v$ assumes the vector $X^v_k$ with probability
$e^{-t}\frac{t^k}{k!}$.  Namely, we consider the combined random walk
\begin{equation*}
X^v=   \sum_{k \leq \K} e^{-t}\frac{t^k}{k!} X^v_k.
\end{equation*}

Now, let $\rho(k)_{S,t,f}$ be the contribution to the heat kernel pagerank
vector $\dirhkpr'$ of walks of length at most $k$.  The expectation of each
$X^v$ is $\rho(k)_{S,t,f}(v)$.  Then, by Lemma~\ref{lem:chernoff},
\begin{align*}
\P(X^v < (1-\epsilon) \rho(k)_{S,t,f}(v)\cdot r) &< \exp(-\rho(k)_{S,t,f}(v)r\epsilon^2/2)\\
&= \exp(-(8/\epsilon)\rho(k)_{S,t,f}(v)\log n)\\
&< n^{-4}
\end{align*}
for every component with $\dirhkpr'(v) > \epsilon$, since then
$\rho(k)_{S,t,f}(v) > \epsilon/2$.  Similarly,
\begin{align*}
\P(X^v > (1+\epsilon) \rho(k)_{S,t,f}(v)\cdot r) &< \exp(-\rho(k)_{S,t,f}(v)r\epsilon^2/4)\\
&= \exp(-(4/\epsilon)\rho(k)_{S,t,f}(v)\log n)\\
&< n^{-2}.
\end{align*}
We conclude the analysis for the support of $\dirhkpr'$ by noting that $\dirhkapprox =
\frac{1}{r} X^v$, and we achieve an $\epsilon$-multiplicative error bound for
every vertex $v$ with $\dirhkpr'(v) > \epsilon$ with probability at least
$1-O(n^{-2})$.

On the other hand, if $\dirhkpr'(v)\leq\epsilon$, by the third part of
Lemma~\ref{lem:chernoff}, $\P(\dirhkapprox(v) > 2\epsilon) \leq n^{-8/\epsilon^2}$.
We conclude that, with high probability, $\dirhkapprox(v) \leq 2\epsilon$.

Finally, when $f$ is not a probability distribution, the above applies to $f' =
f/\onenorm{f}$.  Let $\dirhkapprox'$ be the output of the algorithm using $f' =
f/\onenorm{f}$ and $\dirhkpr'$ be the corresponding Dirichlet heat kernel
pagerank vector $\rho_{S,t,f'}$.  The full error of the Dirichlet heat kernel pagerank returned
is
\begin{align*}
\onenorm{\dirhkpr - \dirhkapprox} &\leq \onenorm{ \onenorm{f}\dirhkpr' -
\onenorm{f}\dirhkapprox' }\\
&\leq \onenorm{f}\onenorm{\dirhkpr' - \dirhkapprox'}\\
&\leq \epsilon\onenorm{f}\onenorm{\dirhkpr'}\\
&= \epsilon\onenorm{f}.
\end{align*}

For the running time, we use the assumptions that performing a random walk step
and drawing from a distribution with finite support require constant time.
These are incorporated in the random walk simulation, which dominates the
computation.  Therefore, for each of the $r$ rounds, at most $K$ steps of the
random walk are simulated, giving a total of $rK =  O\Big(\rparamdir \cdot
\K\Big)= \tilde {O}(t)$ queries.
\end{proof}

\section{The \greensalg~Algorithm}
\label{sec:efficientsolver}
Here we present the main algorithm, \greensalg, for computing a solution to a
Laplacian linear system with a boundary condition.  It is the
\localsolver~algorithmic framework combined with the scheme for approximating
Dirichlet heat kernel pagerank.  The scheme is an optimized version of the
algorithm \dirhkpralg~with a slight modification.  We call the optimized version
\solverapproxhkpr.

\begin{definition}
Define \solverapproxhkprparams~to be the algorithm \dirhkpralgparams~with the
following modification to lines \ref{line:boundk} and~\ref{line:boundk2} after
drawing $k\sim Poiss(t)$:
\begin{quote}
$k \gets min\{k, 2t\}$.
\end{quote}
Namely, this modification limits the length of random walk steps to at most $2t$.
\end{definition}

\begin{theorem}
\label{thm:solverapproxhkpr}
Let $G$ be a graph and $S$ a subset of size $s$.  Let $T = \tparamsolver$, and
let $N = T/\gamma$ for some $0 < \gamma < 1$.  Suppose $j$ is a random variable
drawn from $[1,\lfloor N \rfloor]$ uniformly at random and let $t = jT/N$.  Then
if $\epsilon \geq \gamma$, the algorithm \solverapproxhkpr~returns a
vector that $\epsilon$-approximates $\dirhkpr$ with probability at least
$1-\epsilon$.  Using the same query assumptions as Theorem~\ref{thm:dirhkpralg},
the running time of \solverapproxhkpr~is $O\left(\epsilon^{-3} t \log n
\right)$.
\end{theorem}

We will use the following Chernoff bound for Poisson random variables.

\begin{lemma}[\cite{mitzenmacherpandc}]
\label{lem:poissontails}
Let $X$ be a Poisson random variable with parameter $t$.  Then, if $x > t$,
\begin{equation*}
\P(X \geq x) \leq e^{x-t-x\log(x/t)}.
\end{equation*}
\end{lemma}

\begin{proof}[Proof of Theorem~\ref{thm:solverapproxhkpr}]
Let $k$ be a Poisson random variable with parameter $t$.  Similar to the proof
of Theorem~\ref{thm:dirhkpralg}, we use Lemma~\ref{lem:poissontails} to reason
that
\begin{align*}
\P(k \geq 2t) &\leq e^{2t-t-2t\log(2t/t)}\\
&= e^{t(1 - 2\log 2)}\\
&\leq \epsilon,
\end{align*}
as long as $t \geq \frac{ \log(\epsilon^{-1}) }{ 1-2\log 2 }$.

Let $E$ be the event that $t < \frac{ \log(\epsilon^{-1}) }{ 1-2\log 2 }$.  The
probability of $E$ is
\begin{align*}
\P\left(jT/N < \frac{ \log(\epsilon^{-1})}{1-2\log 2}\right) &=
\P\left(j < \frac{ \log(\epsilon^{-1}) }{ \gamma(1 - 2\log 2) }\right)\\
&= \frac{ \log(\epsilon^{-1}) }{ (1-2\log 2) s^3 \log(s^3 \gamma^{-1}) },
\end{align*}
which is less than $\epsilon$ as long as $\epsilon \geq \left( \frac{ \gamma }{ s^3
} \right)^{(1-2\log 2)\epsilon s^3}$.  This holds when $\epsilon \geq \gamma$.

As before, the algorithm consists of $r$ rounds of random walk simulation, where
each walk is at most $2t$.  The algorithm therefore makes $r\cdot 2t =
\epsilon^{-3} 32t \log n$ queries, requiring $O\left(\epsilon^{-3} t \log n
\right)$ time.
\end{proof}

Below we give the algorithm \greensalg.  The algorithm is identical to
\localsolver~with the exception of line~\ref{line:solverhkpr}, where we use the
approximation algorithm \solverapproxhkpr~for Dirichlet heat kernel
pagerank computation.

\begin{algorithm}
\floatname{algorithm}{Algorithm}
\caption{\greensalgparams}
\label{alg:approxgreen}
input: graph $G$, boundary vector $b\in\R^V$, subset $S\subset V$, solver error
parameter $0<\gamma<1$, Dirichlet heat kernel pagerank error parameter
$0<\epsilon<1$.\\
output: an approximate local solution $\mathtt{x}$ to the local system $x_S=\G
b_1$ satisyfing the boundary condition $b$.\\
\begin{algorithmic}[1]
  \State $s\gets |S|$
  \State initialize a $0$-vector $\mathtt{x}$ of dimension $s$
  \State $b_1 \gets D_S^{-1/2}A_{S, \delta S}D_{\delta S}^{-1/2}b_{\delta S}$
  \State $b_2 \gets b_1^TD_S^{1/2}$
  \State $T \gets \tparamsolver$
  \State $N \gets T/\gamma$
  \State $r \gets \rparamsolver$
  \For{$i=1$ to $r$}
    \State draw $j$ from $[1,N]$ uniformly at random
    \State $x_i \gets$
\solverapproxhkpr($G,jT/N,b_2,S,\epsilon$)\label{line:solverhkpr}
    \State $\mathtt{x} \gets \mathtt{x} + x_i$
  \EndFor\\
  \Return $T/r \cdot \mathtt{x}D_S^{-1/2}$
\end{algorithmic}   
\end{algorithm}

\begin{theorem}
\label{thm:greensalg}
Let $G$ be a graph and $\L$ denote the normalized Laplacian of $G$.  Let $b$ be
a nontrivial vector $b\in \R^V$ and $S$ a $b$-boundable subset, and let $b_1 =
D_S^{-1/2}A_{S, \delta S}D_{\delta S}^{-1/2}b_{\delta S}$.  For the linear
system $\L x = b$, the solution $x$ is required to satisfy the boundary
condition $b$, and let $x_S$ be the local solution.  Then the approximate
solution $\mathtt{x}$ output by the algorithm \greensalg~satisfies the
following:
\begin{enumerate}[(i)]
\item The error of $\mathtt{x}$ is $\norm{x_S-\mathtt{x}} = O\big(
\gamma(\norm{b_1} + \norm{x_S}) + \epsilon\onenorm{b_2} \big)$ with probability
at least $1-\gamma$,\label{greensalgacc}
\item The running time of \greensalg~ is $\greenscomplexity$ where the big-O
constant reflects the time to perform a random walk step, plus additional
preprocessing time $O(|\partial(S)|)$, where $\partial(S)$ denotes the edge
boundary of $S$.\label{greensalgcomp}
\end{enumerate}
\end{theorem}

\begin{proof}
The error of the algorithm using true Dirichlet heat kernel pagerank vectors is
$O\big(\gamma(\norm{b_1} + \norm{x_S})\big)$ by Corollary~\ref{cor:sampler}, so
to prove (\ref{greensalgacc}) we address the additional error of vectors output
by the approximation of \solverapproxhkpr.  By
Theorem~\ref{thm:solverapproxhkpr}, \solverapproxhkpr~outputs an
$\epsilon$-approximate Dirichlet heat kernel pagerank vector with probability at
least $1-\epsilon$.  Let $\dirhkapprox$ be the output of an arbitrary run of
\solverapproxhkprparams.  Then $\norm{\dirhkpr-\dirhkapprox} \leq
\epsilon(\onenorm{\rho_{S,t,f'}} + \onenorm{f}) = \epsilon\onenorm{f}$ by the
definition of $\epsilon$-approximate Dirichlet heat kernel pagerank vectors,
where $f' = f/\onenorm{f}$ is the normalized vector $f$.  This means that the
total error of \greensalg~is
\begin{equation*}
\norm{x_S - \mathtt{x}} \leq O\left(\gamma(\norm{b_1} + \norm{x_S})\right) +
\epsilon\onenorm{b_2}.
\end{equation*}

Next we prove (\ref{greensalgcomp}).  The algorithm makes $r = \rparamsolver$
sequential calls to \solverapproxhkpr.  The maximum possible value of $t$ is $T
= \tparamsolver$, so any call to \solverapproxhkpr~is bounded by $O\left(
\epsilon^{-3} s^3 \log(s^3 \gamma^{-1}) \log n\right)$.  Thus, the total running
time is $\greenscomplexity$.

The additional preprocessing time of $O(|\partial(S)|)$ is for
computing the vectors $b_1$ and $b_2$; these may be computed as a preliminary
procedure.
\end{proof}

We note that the running time above is a sequential running time attained by
calling \solverapproxhkpr~$r$ times.  However, by calling these in $r$ parallel
processes, the algorithm has a parallel running time which is simply the same as
that for \solverapproxhkpr.

\subsection{Restricted Range for Approximation}
Since \solverapproxhkpr~only promises approximate values for vertices whose true
Dirichlet heat kernel pagerank vector values are greater than $\epsilon$, the
\newline\greensalg~algorithm can be optimized even further by preempting when this is
the case.

Figure~\ref{fig:dolphins1norm} illustrates how vector values drop as $t$ gets
large.  The network is the same example network given in
Section~\ref{sec:laplacianlinearsystem} and is further examined in the next
section.  We let $t$ range from $1$ to $T = \tparamsolver \approx 108739$ for
$\gamma = 0.01$ and compute Dirichlet heat kernel pagerank vectors $\dirhkpr$.
The figure plots $L_1$ norms of the vectors as a solid line, and the absolute
value of the maximum entry in the vector as a dashed line.  In this example, no
vector entry is larger than $0.01$ for $t$ as small as $250$.

\begin{figure}
\captionsetup{width=0.8\textwidth} \centering
\includegraphics[width=0.8\textwidth]{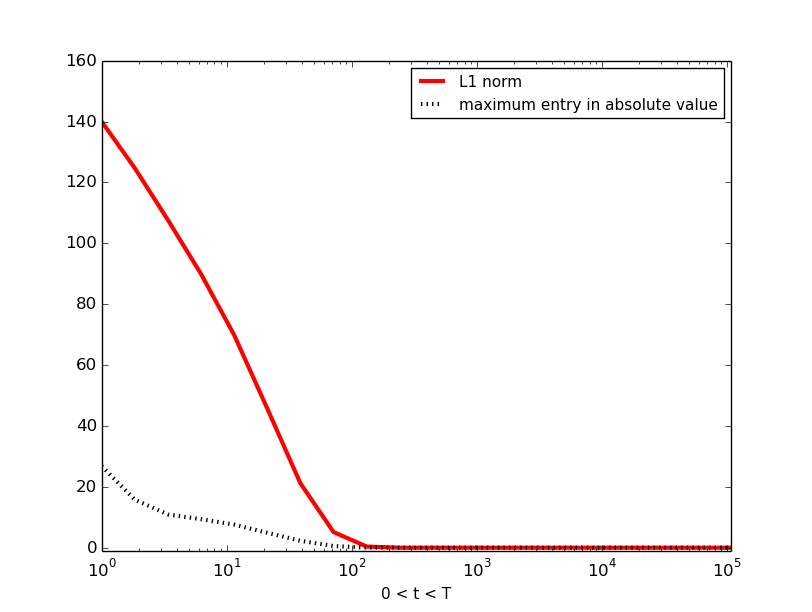} \caption{How support values
of a Dirichlet heat kernel pagerank vector change for different values of $1
\leq t \leq T=\tparamsolver$.  The solid line is the $L_1$ norm -- the sum of
all the support values -- and the dashed line is the absolute value of the
maximum entry in the vector.  Note the x-axis is log-scale.}
\label{fig:dolphins1norm}
\end{figure}

Suppose it is possible to know ahead of time whether a vector $\dirhkpr$ will
have negligably small values for some value $t$.  Then we could skip the
computation of this vector and simply treat it as a vector of all zeros.

From (\ref{eq:error}), the norm of Dirichlet heat kernel pagerank vectors are
monotone decreasing.  Then it is enough to choose a threshold value $t'$ beyond
which $\onenorm{\rho_{S,t',f}} < \epsilon$, since any $\epsilon$-approximation
will return all zeros, and treat this as a cutoff for actually executing the
algorithm.  An optimization heuristic is to only compute
$\solverapproxhkprparams$ if $t$ is less than this threshold value $t'$.
Otherwise we can add zeros (or do nothing).  That is, replace
line~\ref{line:solverhkpr} in \greensalg~with the following:

\begin{algorithmic}
\If{$jT/N < t'$}
  \State $x_i \gets$ \solverapproxhkpr($G,jT/N,b_2,S,\epsilon$)
\Else
  \State do nothing
\EndIf
\end{algorithmic}
From (\ref{eq:error}), a conservative choice for $t'$ is $\frac{1}{\lambda_1}
\log(\epsilon^{-1})$.

\section{An Example Illustrating the Algorithm}
\label{sec:example}
We return to our example to illustrate a run of the Green's solver algorithm for
computing local linear solutions.  The network is a small communication network
of dolphins~\cite{dolphins}.

In this example, the subset has a good cluster, which makes it a good
candidate for an algorithm in which computations are localized.  Namely, it is
ideal for \solverapproxhkpr, which promises good approximation for vertices that
exceed a certain support threshold in terms of the error parameter $\epsilon$.
The support of the vector $b$ is limited to the set of leaders, which is the
vertex boundary of the subset of followers, $l = \delta(f)$.  The vector is
plotted over the agents (vertices) in Figure~\ref{fig:dolphins_boundary_vector}.

\begin{figure}
\begin{centering}
\captionsetup{width=0.8\textwidth}
\makebox[\textwidth][c]{\includegraphics[width=1.2\textwidth]{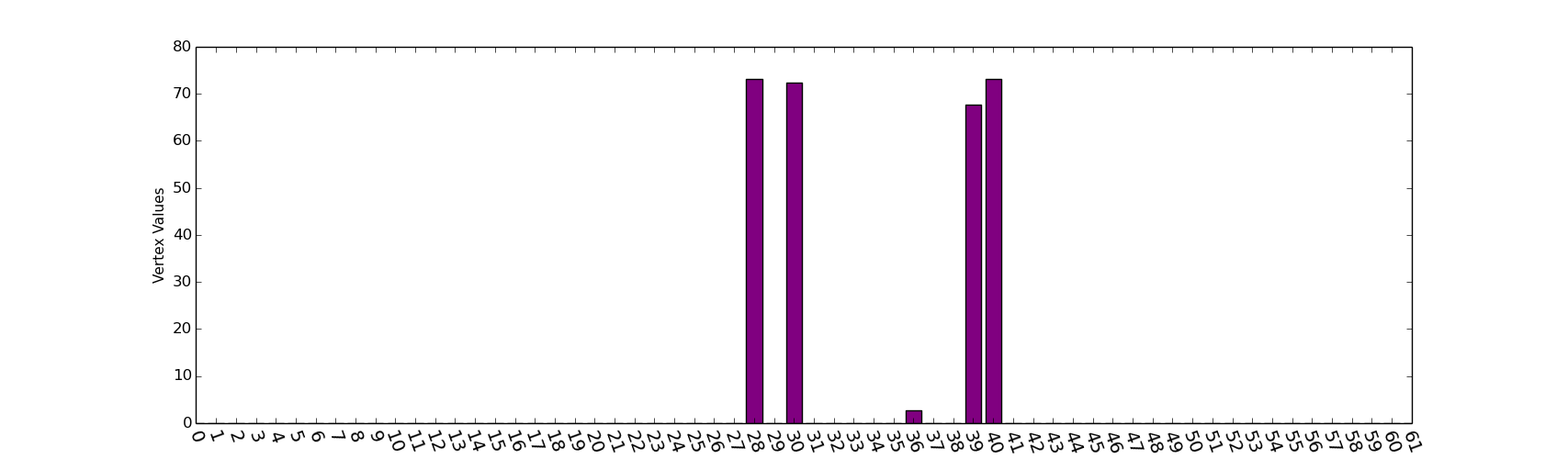}}
\caption{The values of the boundary vector plotted against the agent IDs given
in Figure~\ref{fig:consensus}.}
\label{fig:dolphins_boundary_vector}
\end{centering}
\end{figure}

Figure~\ref{fig:dolphins_sample_hkpr_vector} plots the vector values of the 
heat kernel pagerank vector $\rho_{t,b_2'}$ over the full set of agents.  Here,
we use $b_2'$, the $n$-dimensional vector:
\begin{equation*}
b_2'(v)=
\begin{cases}
b_2(v) \mbox{ if $v\in S$,}\\
0 \mbox{ otherwise,}
\end{cases}
\end{equation*}
and $t=50.0$. The components with largest absolute value are concentrated in the
subset of followers over which we compute the local solution.  This indicates
that an output of \solverapproxhkpr~will capture these values well.

\begin{figure}
\captionsetup{width=0.8\textwidth}
\centering
\makebox[\textwidth][c]{\includegraphics[width=1.2\textwidth]{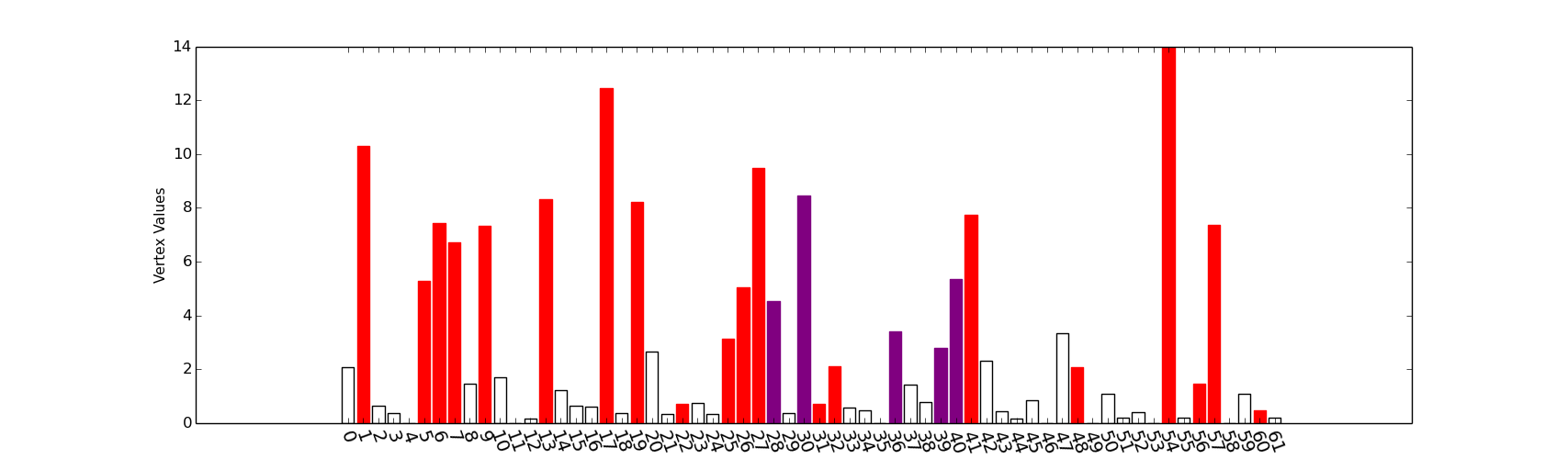}}
\caption{The node values of the full example communication network over a sample heat
kernel pagerank vector.  The red bars correspond to the network of followers,
the purple to the leaders, and the white to the rest of the network.}
\label{fig:dolphins_sample_hkpr_vector}
\end{figure}

\subsection{Approximate solutions}
In the following figures, we plot the results of calls to our approximation
algorithms against the exact solution $x_S$ using the boundary vector of
Figure~\ref{fig:dolphins_boundary_vector}.  The solution $x_S$ is computed by
Theorem~\ref{thm:xS}, and the appromimations are sample outputs of \localsolver
and \greensalg, respectively.  The exact values of $x_S$ are represented by
circles, and the approximate values by triangles in each case.  Note that we
permute the indices of the vertices in the solutions so that vector values in
the exact solution, $x_S$ are decreasing, for reading ease.\footnote{The results
of these experiments as well as the source code are archived at\newline
\url{http://cseweb.ucsd.edu/~osimpson/localsolverexample.html}.}

The result of a sample call to \localsolver~with error parameter $\gamma = 0.01$
is plotted in Figure~\ref{fig:local_linear_solver}.  The total relative error of
this solution is $\frac{\norm{x_S - \hat{x}_S}}{\norm{x_S}} = 0.02$, and the
absolute error $\norm{x_S - \hat{x}_S}$ is within the error bounds given in
Theorem~\ref{thm:localsolver}.  That is, $\norm{x_S - \hat{x}_S} \leq
\gamma\left( \norm{b_1} + \norm{x_S} + \norm{x_{rie}} \right)$, where $x_{rie}$
is the solution obtained by computing the full Riemann sum (as in
Lemma~\ref{lem:sum}).

\begin{figure}
\captionsetup{width=0.8\textwidth}
\centering
\makebox[\textwidth][c]{\includegraphics[width=1.2\textwidth]{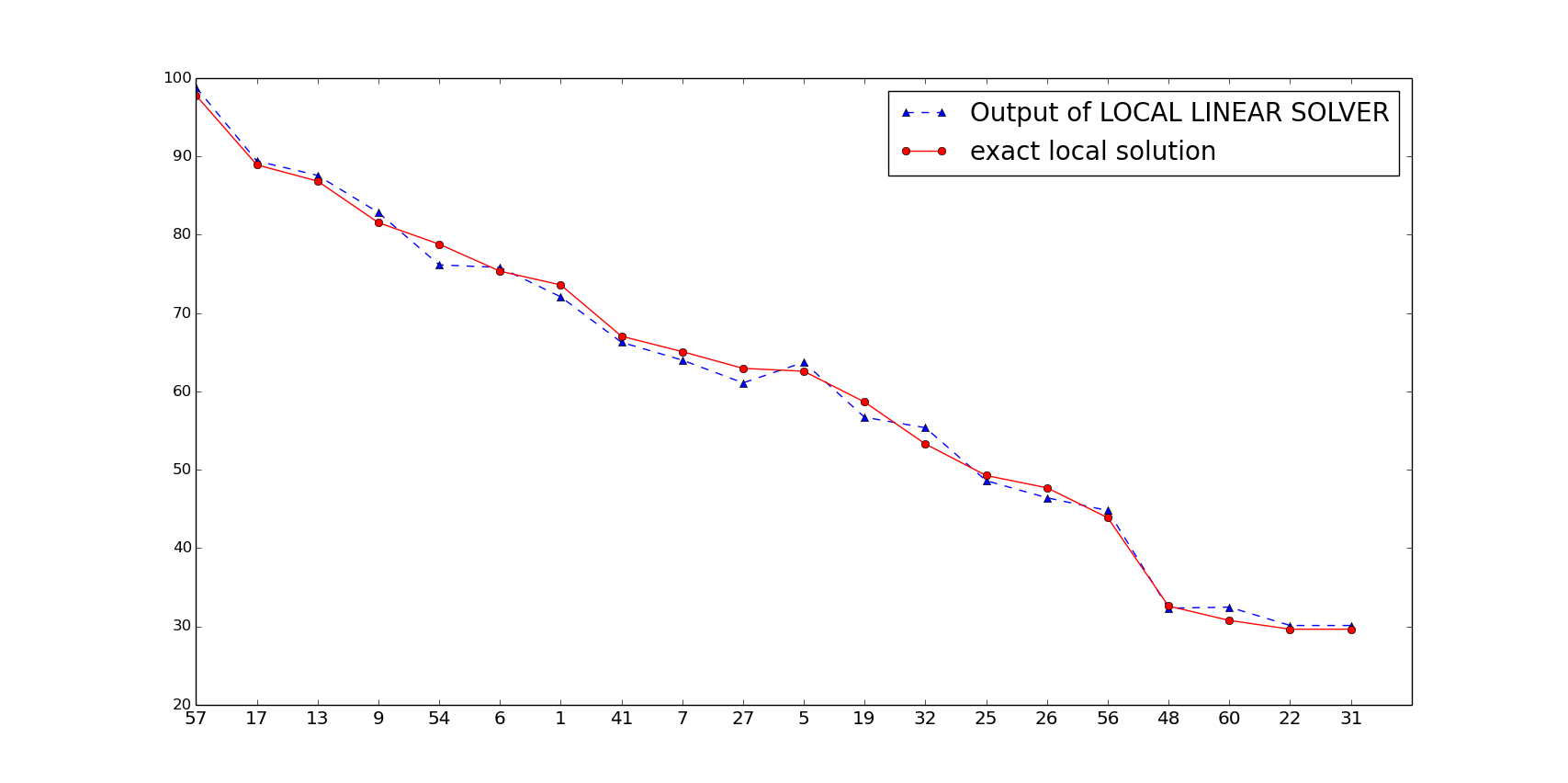}}
\caption{The results of a run of \localsolver.  Two vectors are plotted over IDs
of agents in the subset.  The circles are exact values of $x_S$, while the
triangles are the approximate values returned by \localsolver.}
\label{fig:local_linear_solver}
\end{figure}

The result of a sample call to \greensalg~with parameters $\gamma=0.01,
\epsilon=0.1$ is plotted in Figure~\ref{fig:greens_alg_eps1}.  In this case the
relative error is $\approx 2.05$, but the absolute error meets the error bounds
promised in Theorem~\ref{thm:greensalg} point (\ref{greensalgacc}).
Specifically,
\begin{align*}
\norm{x_S - \hat{x}_S} \leq \left( \gamma( \norm{b_1} + \norm{x_S} + \norm{x_{rie}}
) + \epsilon\onenorm{b_2} \right).
\end{align*}

\begin{figure}
\captionsetup{width=0.8\textwidth}
\centering
\makebox[\textwidth][c]{\includegraphics[width=1.2\textwidth]{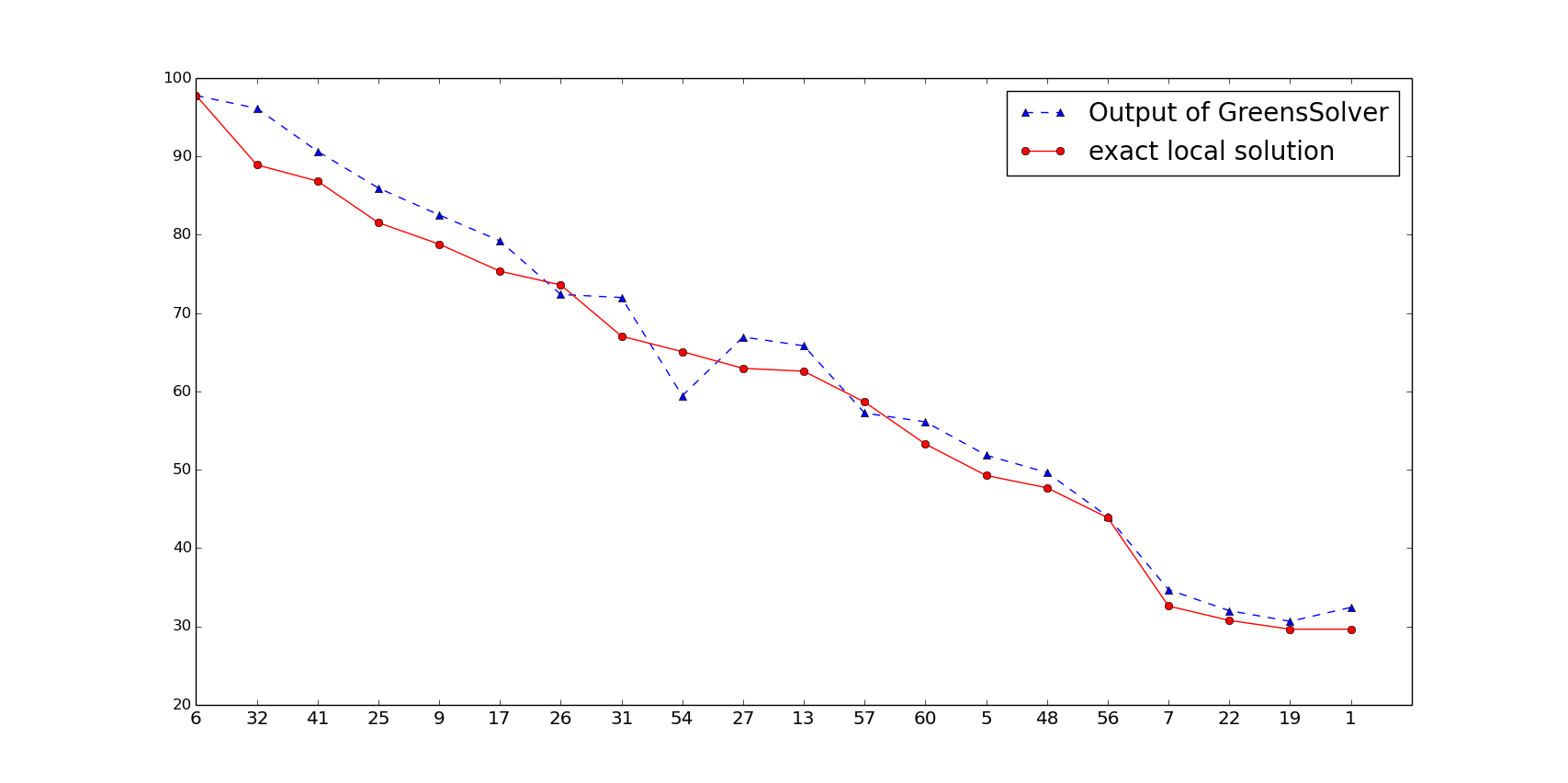}}
\caption{The results of a run of \greensalg~with $\gamma=0.01, \epsilon=0.1$.}
\label{fig:greens_alg_eps1}
\end{figure}

\paragraph{\normalfont\textbf{General remarks.}}While we have focused our
analysis on solving local linear systems with the normalized Laplacian $\L$ as
the coefficient matrix, our methods can be extended to solve local linear
systems expressed in terms of the Laplacian $L$ as well.  There are numerous
applications involving solving such linear systems.  Some examples are discussed
in~\cite{cs:hklinear:13}, and include computing effective resistance in
electrical networks, computing maximum flow by interior point methods,
describing the motion of coupled oscillators, and computing state in a network
of communicating agents.  In addition, we expect the method of approximating
Dirichlet heat kernel pagerank in its own right to be useful in a variety of
related applications.

\paragraph{\normalfont\textbf{Acknowledgements.}} The authors would like to
thank the anonymous reviewers for their comments and suggestions.  Their input has
been immensely helpful in improving the presentation of the results and clarifying
details of the algorithm.

{\footnotesize
\bibliographystyle{amsplain}
\bibliography{im_hklinear_r1}
}

\end{document}